\documentclass[conference, 12pt, onecolumn]{ieeeconf}

\IEEEoverridecommandlockouts
\usepackage[usenames]{color}
\usepackage{enumerate}
\usepackage{url}
\usepackage{subfigure}
\usepackage{amsfonts,mathrsfs}
\usepackage{amssymb,amsmath}
\usepackage{verbatim}
\usepackage{acronym}
\usepackage{mathtools}
\usepackage{cite}
\usepackage{graphicx}
\usepackage{algorithm}
\usepackage[noend]{algpseudocode}

\def\fskip#1{}

\newtheorem{theorem}{Theorem}

\newtheorem{lemma}{Lemma}

\newtheorem{proposition}[theorem]{Proposition}
\newtheorem{remark}{Remark}

\renewcommand{\theenumi}{\arabic{enumi}}

\def\1{{\bf 1}}

\newcommand{\remove}[1]{}

\def\argmin{\mathop{\rm argmin}}

\begin{document}
\title{Distributed Computation for the Non-metric Data Placement Problem using Glauber Dynamics and Auctions}
\author{\authorblockN{S. Rasoul Etesami*\vspace{-0.5cm}}
\thanks{Department of Industrial and Systems Engineering and Coordinated Science Laboratory, University of Illinois Urbana-Champaign,  Urbana, IL 61801 (Email: etesami1@illinois.edu).  This work is supported by the NSF CAREER Award under Grant No. EPCN-1944403.}
}
\maketitle
\begin{abstract}
We consider the non-metric data placement problem and develop distributed algorithms for computing or approximating its optimal integral solution. We first show that the non-metric data placement problem is inapproximable up to a logarithmic factor. We then provide a game-theoretic decomposition of the objective function and show that natural Glauber dynamics in which players update their resources with probability proportional to the utility they receive from caching those resources will converge to an optimal global solution for a sufficiently large noise parameter. In particular, we establish the polynomial mixing time of the Glauber dynamics for a certain range of noise parameters. Finally, we provide another auction-based distributed algorithm, which allows us to approximate the optimal global solution with a performance guarantee that depends on the ratio of the revenue vs. social welfare obtained from the underlying auction. Our results provide the first distributed computation algorithms for the non-metric data placement problem.   
\end{abstract}
\begin{keywords}
Data placement; Glauber dynamics; potential games; approximation algorithms; distributed computation; linear programming duality.
\end{keywords}

\section{Introduction}

Data placement is one of the fundamental resource allocation problems in storage-capable distributed systems, such as content delivery networks, peer-to-peer networks, and mobile networks, for improving system availability, reliability, and fault tolerance \cite{guo2016algorithmic}. The data placement problem looks at how to store copies of different data (resources) among a set of capacity-constrained servers (agents) to minimize the overall resource placement and access costs. The placement cost captures the cost of allocating a particular resource to an agent (e.g., due to data compatibility with a server). On the other hand, the access costs measure the cost of getting access to data points across servers (e.g., due to delay or shipping costs). However, most work on data placement problems assumes that all agents fully comply with the centralized designed protocols. Nevertheless, in real-world data replication applications, entities, e.g., servers, data providers, or data consumers, can belong to different stakeholders or administrative domains with different preferences and objectives \cite{guo2016algorithmic,gopalakrishnan2012cache,etesami2020complexity,etesami2017price}. Therefore, our main goal in this work is to analyze the data placement problem from a distributed computation perspective and without any metric assumption on access costs.

\subsection{Related Work}
The data placement problem has been extensively studied in the past literature. The optimal data placement on networks with a constant number of clients and arbitrary access costs was considered in \cite{angel2014optimal}, where a polynomial algorithm for computing the optimal allocation for uniform resource length was developed. The data placement problem was studied from a mechanism design perspective in \cite{guo2016algorithmic}. There have been several efforts to obtain constant factor approximation algorithms for the \emph{metric} data placement problem, starting with \cite{baev2001approximation} and improved by \cite{baev2008approximation,swamy2016improved,krishnaswamy2018constant}. The main idea behind most of these approximation algorithms is based on solving a natural LP relaxation of the problem and then rounding the solution using refined clustering, network flow, or iterative rounding. A generalization of these results to the so-called ``matroid median problem" has been studied in \cite{swamy2016improved,krishnaswamy2018constant}; as a special case, it results in improved approximation algorithms for the metric data placement problem. We refer to \cite{ansari2017large,thakral2017approximation} for other heuristic approximation algorithms with or without theoretical performance guarantees. It was shown in \cite{etesami2020complexity} that in the case of homogeneous metric data placement when agents have identical request rates, a simple greedy algorithm could achieve an approximation factor of 3, hence substantially improving the existing approximation factors that were given for the heterogeneous setting. In particular, it was shown that the same algorithm results in a 3-competitive algorithm for the online version of the problem in which agents arrive adversarially over time and reveal their specifications. A different online variant of the data placement problem has also been studied in \cite{drwalcompetitive}.

The data placement problem is also closely related to the uncapacitated facility location problem (UFLP) \cite{jain2003greedy} and its variants, such as the $k$-median and matroid median problems \cite{charikar1999constant,deng2022constant}, in which the goal is to open a subset of facilities and assign each client to an open facility in order to minimize the total facility opening costs and clients' access costs. In fact, as we will show, the data placement problem is a more complicated version of the UFLP in which there are multiple facility types that are coupled through the cache constraints. In particular, one can show that by relaxing the cache constraints in the data placement problem using Lagrangian multipliers, the data placement problem can be decomposed into a sum of separable UFLPs. A heuristic approximation algorithm based on decomposing the data placement objective function using Lagrangian relaxation has been studied in \cite{drwal2014decomposition}. 

Unlike most of the past literature, in this work, we look at the non-metric data placement problem and devise distributed computation algorithms to obtain or approximate its global optimum solution. To that end, we provide two distributed algorithms, one based on Glauber dynamics in which the agents best respond (with some noise) to the resource allocation of the other agents, and one based on the first-price auction in which agents sell their cache spaces to bidders who represent different resource types. In both settings, we establish theoretical guarantees on the final resource allocation outcomes and provide new insights on solving the non-metric data placement problem more efficiently in a distributed manner.   

\subsection{Notations} 

We adopt the following notations throughout the paper. For a positive integer $n$, we let $[n]=\{1,2,\ldots,n\}$. For a discrete set $X\subseteq [n]$ and $i\in X$, we often write $X-i$ and $X+i$ to denote $X\setminus\{i\}$ and $X\cup \{i\}$, respectively. For two probability distributions $\mu$ and $\nu$ supported over a finite set $A$, we let $\|\mu-\nu\|_{TV}=\frac{1}{2}\sum_{a\in A}|\mu(a)-\nu(a)|$ be the total variation distance between those distributions. To denote the mixing time of a Markov chain with transition probability matrix $P$, we use $t_{\rm mix}(\epsilon)=\min\{t: d(t)<\epsilon\}$, where $d(t)=\sup_{\mu}\|\mu P^t-\pi\|_{TV}$ is the maximum total variation between the distribution of the Markov chain at time $t$ and its stationary distribution $\pi$. The $\ell_1$ norm of a vector $x$ is denoted by $\|x\|_1=\sum |x_i|$. Given two vectors $x$ and $y$, we let $\rho(x,y)$ be the number of coordinates for which those two vectors differ. We let $\boldsymbol{1}$ and $\boldsymbol{0}$ be the vectors with all one entries and all zero entries, respectively. Finally, for a real number $a$, we let $(a)^+=\max\{a,0\}$.

\section{Problem Formulation and Preliminary Results}
Let us consider the data placement problem \cite{baev2001approximation,swamy2016improved} in which there are a set of $[n]=\{1,2,\ldots,n\}$ agents and a set $[k]=\{1,2,\ldots,k\}$ of unit size resource types. We assume that there are unlimited copies of each resource type. Agent $i\in [n]$ has a cache of size $u_i\in \mathbb{Z}_+$, and the cost of storing resource $\ell$ in its cache is given by $f_i^{\ell}\ge 0$. Moreover, we assume that agent $i\in [n]$ has a nonnegative demand rate $w^{\ell}_i\ge 0$ for recourse $\ell$. In addition, we let $c_{ij}=c_{ji}\ge 0$ be the cost of getting access to agent $i$ from agent $j$.\footnote{Note that unlike \cite{baev2001approximation},\cite{swamy2016improved}, and \cite{etesami2020complexity}, we do not assume anything about the metric property of the access costs.} The goal is to fill agents' caches with proper resources to minimize the overall placement and access costs. More precisely, given an allocation of resources to the agents, let us use $X^{\ell}\subseteq [n]$ to denote the set of agents that hold resource $\ell$ in their caches. Then, the cost of agent $j$ to get access to resource $\ell$ is given by $w^{\ell}_{j}d(j,X^{\ell})$, where $d(j,X^{\ell})=\min\{c_{ij}: i\in X^{\ell}\}$ is the minimum distance that agent $j$ needs to travel to get access to resource $\ell$. In particular, the overall access cost among all the agents is given by $\sum_{j,\ell}w^{\ell}_{j}d(j,X^{\ell})$. An integer program (IP) formulation for the data placement problem is given by
\begin{align}\label{eq:primal}
\min &\sum_{i,j,\ell}w_j^{\ell}c_{ij}x^{\ell}_{ij}+\sum_{i,\ell}f_i^{\ell}y_{i}^{\ell}\cr
& x^{\ell}_{ij}\leq y_i^{\ell} \ \forall i,j,\ell,\cr
& \sum_{i=1}^n x^{\ell}_{ij}\ge 1 \ \forall j,\ell,\cr
& \sum_{\ell=1}^{k}y_{i}^{\ell}\leq u_i \ \forall i,\cr 
& x^{\ell}_{ij}, y^{\ell}_{i}\in \{0,1\}, \ \forall i,j,\ell,
\end{align} 
where $y^{\ell}_i=1$ if we allocate resource $\ell$ to agent $i$, and $x^{\ell}_{ij}=1$ if agent $j$ gets access to resource $\ell$ through agent $i$. The first set of constraints ensures that agent $j$ can access resource $\ell$ through agent $i$ only if agent $i$ holds resource $\ell$ in its cache. The second set of constraints implies that each agent $j$ has to get access to all the resources, and the last set of constraints is the cache capacity constraints that allow agent $i$ to hold at most $u_i$ resources in its cache. Subject to these constraints, the goal is to allocate the resources to the agents to minimize the sum of the placement cost $\sum_{i,\ell}f_i^{\ell}y_{i}^{\ell}$ and the access cost $\sum_{i,j,\ell}w_j^{\ell}c_{ij}x^{\ell}_{ij}$.    

\begin{lemma}
An instance of the data placement problem with arbitrary cache size can be reduced to a unit cache size instance by replacing each agent $i$ with cache size $u_i$, demand vector $w^i=(w_i^{1},\ldots,w_i^{k})$, and installment vector $f_i=(f_i^1,\ldots,f_i^{k})$ with $u_i$ identical agents $i_1,\ldots,i_{u_i}$ with demand vector $\frac{1}{u_i}w^i$ and installment vector $f_i$.
\end{lemma}
\begin{proof}
Consider an arbitrary allocation profile $X=(X^{\ell}, \ell\in [k])$ in the original instance and assume that the content of cache $i$ is filled with resources $\ell_1,\ldots,\ell_{u_i}$. Now let us replace agent $i$ with $u_i$ collocated unit cache size agents, where the caches of agents $i_1,\ldots,i_{u_i}$ are filled with $\ell_1,\ldots,\ell_{u_i}$, respectively. For any agent $j\neq i$, the cost of getting $j$ access to all the resources is the same in both instances. Moreover, the cost of agent $i$ in the original instance equals $\sum_{\ell}w_{i}^{\ell}d(i, X^{\ell})+\sum_{r=1}^{u_i}f_i^{\ell_r}$, while the total cost of agents $i_1,\ldots,i_{u_i}$ in the new instance equals 
\begin{align}\nonumber
\sum_{\ell}\sum_{r=1}^{u_i}  w^{\ell}_{i_r}d(i_r, X^{\ell})+\sum_{r=1}^{u_i}f_i^{\ell_r}=\sum_{\ell}\sum_{r=1}^{u_i}\frac{1}{u_i}w^{\ell}_{i}d(i, X^{\ell})+\sum_{r=1}^{u_i}f_i^{\ell_r}=\sum_{\ell}w_{i}^{\ell}d(i, X^{\ell})+\sum_{r=1}^{u_i}f_i^{\ell_r},
\end{align}
which shows that the total costs of the two instances are the same. The proof is completed by repeating the same process for every agent $i$ until all the agents have unit cache size. 
\end{proof}

In fact, the above reduction holds even for metric access costs, i.e., when $c_{ik}\leq c_{ij}+c_{jk}, \forall i,j,k$. The reason is that the distances in the unit cache size instance still satisfy the metric property, as the distance between any two agents is the same as the distance between their collocated copies. Henceforth, we will focus only on the data placement problem with unit-cache size. But before we get into the analysis, in the following proposition, we show that even approximating the non-metric data placement problem to within a logarithmic factor is a hard problem.

\begin{proposition}\label{prop_NP-hard}
It is NP-hard to approximate the non-metric data placement problem up to a factor better than $O(\ln n)$.
\end{proposition} 
\begin{proof}
We show that the non-metric uncapacitated facility location problem can be formulated as a special instance of the non-metric data placement problem. On the other hand, it is known that the non-metric UFLP with $n$ clients is at least as hard as the set cover problem, which is hard to approximate within an $O(\ln n)$ factor \cite{hochbaum1982heuristics,bienkowski2020nearly}. Therefore, the same inapproximability result must also hold for the non-metric data placement problem. 

Consider an arbitrary instance of the UFLP with the same set $[n]$ of clients and facilities, non-metric access costs $\{c_{ij}: i,j\in [n]\}$, and facility installment costs $\{f_i, i\in [n]\}$. This problem can be formulated as an instance of the data placement problem with a set $[n]$ of agents, non-metric access costs $\{c_{ij}: i,j\in [n]\}$, and $k=2$ resources. For the first resource we set $w^{1}_j=1, f^{1}_i=f_i,  \forall i,j\in [n]$. For the second (dummy) resource we set $w^{2}_j=0, f^{2}_i=0,  \forall i,j\in [n]$. In other words, the agents that receive resource $\ell=1$ correspond to the set of open facilities in the UFLP, while the agents that receive the dummy resource $\ell=2$ correspond to the set of closed facilities. By the construction, it should be clear that any optimal solution to the UFLP corresponds to an optimal solution in the data placement problem with the same objective cost, and vice versa.   
\end{proof}

Despite the above negative result, we are still interested in finding distributed algorithms that perform well in most instances of the non-metric data placement problem. To that end, we will develop two distributed algorithms in which either the agents or the resources are viewed as selfish entities that aim to maximize their payoffs, and we analyze the performance of the allocation profiles resulting from agents' interactions.

\section{A Game-Theoretic Decomposition for the Data Placement Problem}

Let us consider the objective function of the unit cache size data placement problem 
\begin{align}\nonumber
\Phi(x)=\sum_{j,\ell}w_j^{\ell}d(j,X^{\ell})+\sum_{i}f_i^{x_i},
\end{align}
where $x=(x_1,\ldots,x_n)\in [k]^n$ denotes the resource allocation profile of all the agents, $X^{\ell}=\{i: x_i=\ell\}$ is the set of agents that have resource $\ell$ in their cache, and $d(j,X^{\ell})=\min\{c_{ji}: i\in X^{\ell}\}$. Consider a noncooperative game in which each agent $i\in [n]$ can be viewed as one player with the action set $[k]$. The action of player $i$ is the resource $x_i\in [k]$ that it caches, and incurs a cost that is given by 
\begin{align}\nonumber
c_i(x)=\sum_{j,\ell}w_j^{\ell}\big(d(j,X^{\ell})-c_{ij}\big)^++f_i^{x_i},
\end{align}
where for a real number $a$ we define $(a)^+=\max\{0,a\}$. 

\begin{lemma}\label{lemm:potential}
The above noncooperative game $\mathcal{G}=([n],[k]^n,\{c_i\})$ is an exact potential game with the potential function $\Phi(x)$. 
\end{lemma}
\begin{proof}
Consider an arbitrary player $i$ and an action (allocation) profile $x=(x_i,x_{-i})$ such that $x_i=\ell$. Assume that player $i$ changes its action from $x_i=\ell$ to $x'_i=\ell'$, and call the new action profile $x'=(x'_i,x_{-i})$. For any $o\in [k]$, let us use $X^o$ and $X'^o$ to denote the set of players holding resource $o$ in action profiles $x$ and $x'$, respectively. Then, we have $i\in X^{\ell}, i\notin X^{\ell'}$ and $X'^{\ell}=X^{\ell}-i, X'^{\ell'}=X^{\ell'}+i$. We can write 
\begin{align}\label{eq:change-in-cost}
c_i(x'_i,x_{-i})-c_i(x_i,x_{-i})&=\sum_j w_j^{\ell}\big(d(j,X'^{\ell})-c_{ij}\big)^++\sum_j w_j^{\ell'}\big(d(j,X'^{\ell'})-c_{ij}\big)^+\cr
&\qquad-\sum_j w_j^{\ell}\big(d(j,X^{\ell})-c_{ij}\big)^+-\sum_j w_j^{\ell'}\big(d(j,X^{\ell'})-c_{ij}\big)^++f_i^{\ell'}-f_i^{\ell}\cr
&=\sum_j w_j^{\ell}\big(d(j,X^{\ell}-i)-c_{ij}\big)^+-\sum_j w_j^{\ell'}\big(d(j,X^{\ell'})-c_{ij}\big)^++f_i^{\ell'}-f_i^{\ell},
\end{align}
where the first equality holds because for any resource $o\notin\{\ell, \ell'\}$, we have $X^o=X'^o$. The second equality follows from $d(j,X^{\ell})\leq c_{ij}$ and $d(j,X'^{\ell'})\leq c_{ij}$ because $i\in X^{\ell}, i\in X'^{\ell'}$.

Next, we compute the amount of change in the potential function $\Phi(x)$. We have
\begin{align}\label{eq:change-in-potential}
\Phi(x'_i,x_{-i})&-\Phi(x_i,x_{-i})\cr 
&=\sum_j w_j^{\ell}d(j,X'^{\ell})+\sum_j w_j^{\ell'}d(j,X'^{\ell'})-\sum_j w_j^{\ell}d(j,X^{\ell})-\sum_j w_j^{\ell'}d(j,X^{\ell'})+f_i^{\ell'}-f_i^{\ell}\cr 
&=\sum_j w_j^{\ell}\big(d(j,X'^{\ell})-d(j,X^{\ell})\big)-\sum_j w_j^{\ell'}\big(d(j,X^{\ell'})-d(j,X'^{\ell'})\big)+f_i^{\ell'}-f_i^{\ell}\cr
&=\sum_j w_j^{\ell}\big(d(j,X^{\ell}-i)-d(j,X^{\ell})\big)-\sum_j w_j^{\ell'}\big(d(j,X^{\ell'})-d(j,X^{\ell'}+i)\big)+f_i^{\ell'}-f_i^{\ell}\cr 
&=\sum_j w_j^{\ell}\big(d(j,X^{\ell}-i)-d(j,X^{\ell})\big)-\sum_j w_j^{\ell'}\big(d(j,X^{\ell'})-c_{ij}\big)^++f_i^{\ell'}-f_i^{\ell}\cr
&=\sum_j w_j^{\ell}\big(d(j,X^{\ell}-i)-c_{ij}\big)^+-\sum_j w_j^{\ell'}\big(d(j,X^{\ell'})-c_{ij}\big)^++f_i^{\ell'}-f_i^{\ell},
\end{align} 
where in the fourth equality we have used the fact that $d(j,X^{\ell'})-d(j,X^{\ell'}+i)=\big(d(j,X^{\ell'})-c_{ij}\big)^+$ by considering two cases. First, if $d(j,X^{\ell'})\leq c_{ij}$, then $d(j,X^{\ell'}+i)=d(j,X^{\ell'})$, and thus $d(j,X^{\ell'})-d(j,X^{\ell'}+i)=0=\big(d(j,X^{\ell'})-c_{ij}\big)^+$. Second, if $d(j,X^{\ell'})> c_{ij}$, then $d(j,X^{\ell'}+i)=c_{ij}$, and thus $d(j,X^{\ell'})-d(j,X^{\ell'}+i)=d(j,X^{\ell'})-c_{ij}=\big(d(j,X^{\ell'})-c_{ij}\big)^+$. Similarly, the last equality is obtained from $d(j,X^{\ell}-i)-d(j,X^{\ell})=\big(d(j,X^{\ell}-i)-c_{ij}\big)^+$, which can be shown by considering two cases: if $d(j,X^{\ell}-i)\leq c_{ij}$, then $d(j,X^{\ell}-i)=d(j,X^{\ell})$, and thus $d(j,X^{\ell}-i)-d(j,X^{\ell})=0=\big(d(j,X^{\ell}-i)-c_{ij}\big)^+$. Otherwise, if $d(j,X^{\ell}-i)> c_{ij}$, then $d(j,X^{\ell})=c_{ij}$, and thus $d(j,X^{\ell}-i)-d(j,X^{\ell})=d(j,X^{\ell}-i)-c_{ij}=\big(d(j,X^{\ell}-i)-c_{ij}\big)^+$. Finally, by comparing \eqref{eq:change-in-cost} and \eqref{eq:change-in-potential}, one can see that $\Phi(x'_i,x_{-i})-\Phi(x_i,x_{-i})=c_i(x'_i,x_{-i})-c_i(x_i,x_{-i})$, which completes the proof.     
\end{proof}

\subsection{Performance Guarantee of a Pure Nash Equilibrium}

As a result of Lemma \ref{lemm:potential}, if players selfishly update their resources by minimizing their cost functions, the overall allocation profile will converge to a pure Nash equilibrium (NE), which must be a local minimum of the potential function. Therefore, one could ask about the quality of the solution obtained at a NE compared to the global optimum of the potential function, which is the optimal solution to the data placement problem. To evaluate the quality of a solution obtained at a NE, we leverage the dual program corresponding to the linear program relaxation of the data placement problem \eqref{eq:primal}, which is given by 
\begin{align}\label{eq:dual}
\max &\sum_{j,\ell}\beta_j^{\ell}-\sum_i \alpha_i\cr
& \beta_j^{\ell}-u_{ij}^{\ell}\leq w_j^{\ell}c_{ij} \ \forall i,j,\ell,\cr
& \sum_j u^{\ell}_{ij}-\alpha_i\leq  f^{\ell}_i \ \forall i,\ell,\cr
& u^{\ell}_{ij}, \beta^{\ell}_{j}, \alpha_i \ge 0, \ \forall i,j,\ell.
\end{align} 
Using the first set of constraints, in an optimal dual solution we may assume $u^{\ell}_{ij}=\big(\beta_j^{\ell}-w_j^{\ell}c_{ij}\big)^+, \forall i,j,\ell$. Otherwise, if $u^{\ell}_{ij}>\big(\beta_j^{\ell}-w_j^{\ell}c_{ij}\big)^+$ for some $i,j,\ell$, we can create a new feasible dual solution by reducing $u^{\ell}_{ij}$ to $\big(\beta_j^{\ell}-w_j^{\ell}c_{ij}\big)^+$. Such a change preserves the dual feasibility of the second set of constraints while potentially allowing one to reduce $\alpha_i$ and hence increase the dual objective value. By abuse of notation, if we use $\beta_{j}^{\ell}$ to denote $\frac{\beta_j^{\ell}}{w_j^{\ell}}$, we can write the dual program \eqref{eq:dual} in an equivalent form as
\begin{align}\label{eq:eq_dual}
\max &\sum_{j,\ell}w_j^{\ell}\beta_j^{\ell}-\sum_i \alpha_i\cr
& \sum_j w_j^{\ell}\big(\beta_j^{\ell}-c_{ij}\big)^+-f^{\ell}_i\leq \alpha_i   \ \forall i,\ell,\cr
& \beta^{\ell}_{j}, \alpha_i \ge 0, \ \forall i,j,\ell.
\end{align}

Next, let us use $x=(x_i,x_{-i})$ to denote a pure NE of the potential game $\mathcal{G}$. Then, for any player $i$ and any action $x'_i$, if we let $x'=(x'_i,x_{-i})$, we must have $c_i(x)\leq c_i(x')$, which implies
\begin{align}\nonumber
\sum_jw_j^{x_i}\big(d(j,X^{x_i})-c_{ij}\big)^+&+\sum_jw_j^{x'_i}\big(d(j,X^{x'_i})-c_{ij}\big)^++f_i^{x_i}\cr 
 &\leq \sum_jw_j^{x_i}\big(d(j,X^{x_i}-i)-c_{ij}\big)^++\sum_jw_j^{x'_i}\big(d(j,X^{x'_i}+i)-c_{ij}\big)^++f_i^{x'_i}.
\end{align}  
Since $i\in X^{x_i}$ and $i\in X^{x'_i}+i$, we have
\begin{align}\nonumber
\sum_jw_j^{x'_i}\big(d(j,X^{x'_i})-c_{ij}\big)^+-f_i^{x'_i}\leq \sum_jw_j^{x_i}\big(d(j,X^{x_i}-i)-c_{ij}\big)^+-f_i^{x_i}, \ \ \forall i, x'_i.
\end{align}
That means that if we define $\beta_{j}^{\ell}=d(j,X^{\ell})\ge 0$ and $\alpha_i=\sum_jw_j^{x_i}\big(d(j,X^{x_i}-i)-c_{ij}\big)^+-f_i^{x_i}$, then $(\alpha_i,\beta^{\ell}_{j})$ forms a feasible dual solution to the dual program \eqref{eq:eq_dual} whose objective value by weak duality is less than the optimal fractional solution to the LP relaxation of \eqref{eq:primal}. Therefore, if the optimal solution of the data placement problem is denoted by $x^o$ with minimum objective cost $\Phi(x^o)$, we have
\begin{align}\nonumber
\Phi(x^o)\ge \sum_{j,\ell}w_j^{\ell}\beta^{\ell}_{j}-\sum_{i}\alpha_i&=\sum_{j,\ell}w_j^{\ell}d(j,X^{\ell})+\sum_if_i^{x_i}-\sum_i \sum_jw_j^{x_i}\big(d(j,X^{x_i}-i)-c_{ij}\big)^+-\sum_if_i^{x_i}\cr 
&=\Phi(x)-\sum_i \sum_jw_j^{x_i}\big(d(j,X^{x_i}-i)-c_{ij}\big)^+.
\end{align}
As a result, the objective value of the solution obtained at NE $x$ is at most
\begin{align}\nonumber
\Phi(x)&\leq \Phi(x^o)+\sum_i \sum_jw_j^{x_i}\big(d(j,X^{x_i}-i)-c_{ij}\big)^+\cr 
&=\Phi(x^o)+\sum_{\ell}\sum_{i\in X^{\ell}} \sum_jw_j^{\ell}\big(d(j,X^{\ell}-i)-c_{ij}\big)^+\cr 
&=\Phi(x^o)+\sum_{\ell}\sum_j w_j^{\ell}\Big(\sum_{i\in X^{\ell}} \big(d(j,X^{\ell}-i)-c_{ij}\big)^+\Big)\cr
&=\Phi(x^o)+\sum_{\ell}\sum_j w_j^{\ell}\big(d(j,X^{\ell}-i_j)-d(j,X^{\ell})\big), 
\end{align}
where $i_j=\argmin_{k\in X^{\ell}} c_{jk}$, and the last equality holds because for any $i\in X^{\ell}-i_j$, we have $d(j,X^{\ell}-i)=c_{ji_j}=d(j,X^{\ell})\leq c_{ij}$, and hence $\big(d(j,X^{\ell}-i)-c_{ij}\big)^+=0$. Thus, if we let $\Phi(x\setminus j)$ be the value of the potential function when the cache content of player $j$ is evacuated, from the above expression we have 
\begin{align}\nonumber
\Phi(x)&\leq \Phi(x^o)+\sum_{\ell}\sum_j w_j^{\ell}\big(d(j,X^{\ell}-i_j)-d(j,X^{\ell})\big)\cr
&= \Phi(x^o)+\sum_j \big(\Phi(x\setminus j)-\Phi(x)\big),
\end{align}
or, equivalently,
\begin{align}\nonumber
\Phi(x)\leq \frac{\Phi(x^o)+\sum_j \Phi(x\setminus j)}{n+1}.
\end{align}
That gives an upper bound for the quality of a NE in terms of the global minimum value and the objective function's sensitivity to each player's cache content at that NE. 

In fact, we believe that in the worst-case scenario, the quality of an arbitrary NE can be significantly smaller than that of a globally optimal solution. The reason is that a NE is the outcome of a local search algorithm that is unimprovable up to a single-player deviation. However, it is known that for the simpler UFLP or $k$-median problem, a richer class of local search moves are required to guarantee the existence of a ``good" suboptimal solution \cite{williamson2011design}. Therefore, in the next section, we rely on Glauber dynamics with noisy updates to steer the resource allocation outcome resulting from players' interactions closer to the global optimal solution.

\section{Glauber Dynamics for Finding A Global Optimal Solution}

Let $\mathcal{X}=[k]^n$ be the space of all possible resource allocations. We consider Glauber dynamics over the space $\mathcal{X}$ in which players iteratively update their cache contents. More precisely, given an allocation profile $x\in \mathcal{X}$, at each time instance $t=1,2,\ldots$, one player $i$ will be chosen uniformly and independently from the past and will update its resource to $o\in [k]$ with probability
\begin{align}\label{eq:gauber_dynamics}
\frac{e^{-\beta c_i(o, x_{-i})}}{\sum_{\ell\in [k]}e^{-\beta c_i(\ell, x_{-i})}},
\end{align}   
where $\beta\in [0, \infty)$ is a noise parameter. In other words, given that player $i$ is chosen to update its resource at time $t$, the probability that it caches resource $o$ is proportional to the utility that resource $o$ brings to that player subject to an additional noise $\beta$ that captures the uncertainty or mistake of player $i$ in choosing resource $o$. As $\beta\to \infty$, the above Glauber dynamics replicate the best response dynamics. Moreover, one can see that the above Glauber dynamics induce a Markov chain over the state space of all the allocation profiles $\mathcal{X}$. The following lemma shows that the stationary distribution of such a Markov chain is given by the Gibbs distribution with respect to the potential function $\Phi$.
\begin{lemma}
The stationary distribution of the Markov chain induced by the Glauber dynamics is given by $\pi:\mathcal{X}\to [0,1]$, where 
\begin{align}\label{eq:gibbs}
\pi(x)=\frac{e^{-\beta \Phi(x)}}{\sum_{z\in \mathcal{X}}e^{-\beta \Phi(z)}}. 
\end{align}
\end{lemma}
\begin{proof}
We first note that any transition of the Markov chain is between two states that differ in the resource of at most one player. We show that the distribution \eqref{eq:gibbs} satisfies the \emph{detailed-balanced} conditions \cite{levin2017markov}, and hence must be a stationary distribution for the induced Markov chain. Let us consider two allocation profiles $x$ and $y$ that differ in the resource of at most one player $i$, that is, $x_{-i}=y_{-i}$. Then, we have
\begin{align}\nonumber
\pi(x)P_{xy}&= \pi(x)\frac{\frac{1}{n}e^{-\beta c_i(y_i, x_{-i})}}{\sum_{\ell\in [k]}e^{-\beta c_i(\ell, x_{-i})}}=\pi(x)\frac{\frac{1}{n}e^{-\beta (c_i(y_i, x_{-i})-c_i(x))}}{\sum_{\ell\in [k]}e^{-\beta (c_i(\ell, x_{-i})-c_i(x))}}\cr 
&=\pi(x)\frac{\frac{1}{n}e^{-\beta (\Phi(y_i, x_{-i})-\Phi(x))}}{\sum_{\ell\in [k]}e^{-\beta (\Phi(\ell, x_{-i})-\Phi(x))}}=\pi(x)\frac{\frac{1}{n}e^{-\beta \Phi(y_i, x_{-i})}}{\sum_{\ell\in [k]}e^{-\beta \Phi(\ell, x_{-i})}}\cr 
&=\frac{1}{n}\big(\frac{e^{-\beta \Phi(x)}}{\sum_{z\in \mathcal{X}}e^{-\beta \Phi(z)}}\big)\big(\frac{e^{-\beta \Phi(y)}}{\sum_{\ell\in [k]}e^{-\beta \Phi(\ell, x_{-i})}}\big).
\end{align}
Similarly, one can show that 
\begin{align}\nonumber
\pi(y)P_{yx}= \frac{\frac{1}{n}e^{-\beta \Phi(y)}}{\sum_{z\in \mathcal{X}}e^{-\beta \Phi(z)}}\frac{e^{-\beta c_i(x_i, x_{-i})}}{\sum_{\ell\in [k]}e^{-\beta c_i(\ell, x_{-i})}}=\frac{1}{n}\big(\frac{e^{-\beta \Phi(x)}}{\sum_{z\in \mathcal{X}}e^{-\beta \Phi(z)}}\big)\big(\frac{e^{-\beta \Phi(y)}}{\sum_{\ell\in [k]}e^{-\beta \Phi(\ell, x_{-i})}}\big).
\end{align}  
Comparing the above two relations shows that $\pi(x)P_{xy}=\pi(y)P_{yx}$, which completes the proof.   
\end{proof}

As a result, for sufficiently large $\beta$, the Glauber dynamics will concentrate on an allocation profile with the smallest potential function, which is the global minimum of the data placement problem. However, for larger $\beta$, the induced chain takes longer to mix to its stationary distribution. However, the following theorem shows that if $\beta$ is not very large, the induced Markov chain mixes quickly to its stationary Gibbs distribution.

\begin{theorem}\label{thm:mixing}
The mixing time of the Glauber dynamics for $\beta\leq \frac{k}{6nu}$ is at most $t_{\rm mix}(\epsilon)=O(n\ln\frac{n}{\epsilon})$, where $u=\max_{i,x}c_i(x)$. 
\end{theorem}
\begin{proof}
Let us consider two allocation profiles $x$ and $y$ that differ in the resource of exactly one player $i$, that is, $x_{-i}=y_{-i}$ and $x_i\neq y_i$. Let $Z^x_t$ and $Z^y_t$ be the Markov chains obtained from Glauber dynamics with initial states $x$ and $y$, respectively. Moreover, by abuse of notation, let us use $x$ and $y$ to denote the current states of the two Markov chains, respectively. Assuming that player $i'\in [n]$ is selected to update its action at the current time, the transition probability distributions of the chains denoted by $\mu^{i'}$ and $\nu^{i'}$ are given by
\begin{align}\nonumber
&\mu^{i'}_o=\frac{e^{-\beta c_{i'}(o, x_{-i'})}}{\sum_{\ell\in [k]}e^{-\beta c_{i'}(\ell, x_{-i'})}},\  o\in [k],\cr
&\nu^{i'}_o=\frac{e^{-\beta c_{i'}(o, y_{-i'})}}{\sum_{\ell\in [k]}e^{-\beta c_{i'}(\ell, y_{-i'})}},\ o\in [k].
\end{align}  
We couple these chains together by allowing the same player and the same resource (whenever possible) to be used in both chains at each time instance. More precisely, if player $i'=i$ is selected to update, then in both chains, we update the resource of player $i$ to $o$ with the probability given in \eqref{eq:gauber_dynamics}. Otherwise, if player $i'\neq i$ is selected, we update the resource of player $i'$ in both chains according to the optimal coupling between the distributions $\mu^{i'}$ and $\nu^{i'}$.\footnote{Given two random variables $X$ and $Y$ with distributions $\pi_X$ and $\pi_Y$, the optimal coupling between them induces a joint probability distribution $\mathbb{P}$ over $(X, Y)$ such that $P(X\neq Y)=\|\pi_X-\pi_Y\|_{TV}$, where $\|\cdot\|_{TV}$ denotes the total variation distance.}

For two action profiles $z,z'\in \mathcal{X}$, let $\rho(z, z')$ denote the number of positions in which $z$ and $z'$ differ from each other. According to the above coupling, when player $i$ is selected, the two chains become identical, i.e., $\rho(Z^{x}_1, Z^{y}_1)=0$. Thus, $\rho(Z^{x}_1, Z^{y}_1)$ might increase from $1$ to $2$ only if a player $i'\neq i$ were selected, and the resource of that player were updated to two different resources in those chains. Let $M^{i'}$ and $N^{i'}$ be the random variables denoting the updated resource of player $i'$ with corresponding distributions $\mu^{i'}$ and $\nu^{i'}$, respectively. We have
\begin{align}\nonumber
\mathbb{P}\{\rho(Z^{x}_1, Z^{y}_1)=2\}=\frac{1}{n}\sum_{i'\neq i}\mathbb{P}\{M^{i'}\neq N^{i'}\}=\frac{1}{n}\sum_{i'\neq i}\|\mu^{i'}-\nu^{i'}\|_{TV},
\end{align}
where the second equality holds because we use optimal coupling of distributions $\mu^{i'}$ and $\nu^{i'}$ to update the resource of player $i'$.

Next, we proceed to bound $\|\mu^{i'}-\nu^{i'}\|_{TV}$. Given action profiles $x=(x_i,x_{-i})$ and $y=(y_i,x_{-i})$, by abuse of notation, let $X^{x_i}$ and $X^{y_i}$ be the set of players in $[n]\setminus\{i'\}$ that hold resources $x_i$ and $y_i$ in the allocation profile $x$, respectively. Then, if $x_{i'}\notin \{x_i,y_i\}$, we have 
\begin{align}\nonumber
c_{i'}(x)-c_{i'}(y)&=\sum_{j}w_{j}^{y_i}\big(d(j,X^{y_i})-c_{i'j}\big)^++\sum_{j}w_{j}^{x_i}\big(d(j,X^{x_i})-c_{i'j}\big)^+\cr 
&-\sum_{j}w_{j}^{y_i}\big(d(j,X^{y_i}+i)-c_{i'j}\big)^+-\sum_{j}w_{j}^{x_i}\big(d(j,X^{x_i}-i)-c_{i'j}\big)^+:=\Delta.
\end{align}
Otherwise, if $x_{i'}=x_i$, then 
\begin{align}\nonumber
c_{i'}(x)-c_{i'}(y)&=\sum_{j}w_{j}^{y_i}\Big(\big(d(j,X^{y_i})-c_{i'j}\big)^+-\big(d(j,X^{y_i}+i)-c_{i'j}\big)^+\Big)=\Delta+\Delta_{x_i},
\end{align}
where $\Delta_{x_i}:=\sum_{j}w_{j}^{x_i}\Big(\big(d(j,X^{x_i}-i)-c_{i'j}\big)^+-\big(d(j,X^{x_i})-c_{i'j}\big)^+\Big)\ge 0$. Similarly, if $x_{i'}=y_i$, we have  
\begin{align}\nonumber
c_{i'}(x)-c_{i'}(y)=\sum_{j}w_{j}^{x_i}\Big(\big(d(j,X^{x_i})-c_{i'j}\big)^+-\big(d(j,X^{x_i}-i)-c_{i'j}\big)^+\Big)=\Delta-\Delta_{y_i},
\end{align}
where $\Delta_{y_i}:=\sum_{j}w_{j}^{y_i}\Big(\big(d(j,X^{y_i})-c_{i'j}\big)^+-\big(d(j,X^{y_i}+i)-c_{i'j}\big)^+\Big)\ge 0$. Now let us define the notations 
\begin{align}\nonumber
&B_1=e^{-\beta c_{i'}(x_i,x_{-i'})}, \ \ \ \ \qquad B_2=e^{-\beta c_{i'}(y_i,x_{-i'})}, \ \ \ \qquad\qquad B=B_1+B_2,\cr
&C_1=e^{-\beta(c_{i'}(x_i,x_{-i'})-\Delta_{x_i})},  \ \ \  C_2=e^{-\beta(c_{i'}(y_i,x_{-i'})+\Delta_{y_i})}, \ \ \qquad C=C_1+C_2,\cr 
&A=\sum_{\ell\notin \{x_i,y_i\}}e^{-\beta c_{i'}(\ell, x_{-i'})},
\end{align}
and note that $B_2\ge C_2$ and $C_1\ge B_1$. Then, the probability distribution $\nu^{i'}$ can be written as
\begin{align}\nonumber
\nu^{i'}_o=\begin{cases}
\frac{\exp(-\beta c_{i'}(o, x_{-i'}))}{A+C} & \mbox{if} \ o\notin\{x_i,y_i\}, \\
\frac{C_1}{A+C} & \mbox{if} \ o=x_i, \\
\frac{C_2}{A+C} & \mbox{if} \ o=y_i.
\end{cases}
\end{align}
Therefore, by definition of the total variation, we have
\begin{align}\label{eq:total-variations}
2\|\nu^{i'}-\mu^{i'}\|_{TV}&=|\nu^{i'}_{x_i}-\mu^{i'}_{x_i}|+|\nu^{i'}_{y_i}-\mu^{i'}_{y_i}|+\sum_{o\notin \{x_i,y_i\}}|\nu^{i'}_{o}-\mu^{i'}_{o}|.
\end{align}
To bound the last term in \eqref{eq:total-variations}, for any $o\notin \{x_i,y_i\}$, we have 
\begin{align}\nonumber
\sum_{o\notin \{x_i,y_i\}}|\nu^{i'}_{o}-\mu^{i'}_{o}|=\sum_{o\notin \{x_i,y_i\}}\big|\frac{e^{-\beta c_{i'}(o, x_{-i'})}}{A+C}-\frac{e^{-\beta c_{i'}(o, x_{-i'})}}{A+B}\big|=\frac{A|B-C|}{(A+C)(A+B)}.
\end{align}
Similarly, we can compute the first two terms in \eqref{eq:total-variations} as
\begin{align}\nonumber
|\nu^{i'}_{x_i}-\mu^{i'}_{x_i}|=\frac{C_1}{A+C}-\frac{B_1}{A+B}, \qquad\qquad |\nu^{i'}_{y_i}-\mu^{i'}_{y_i}|=\frac{B_2}{A+B}-\frac{C_2}{A+C}.
\end{align}
Substituting the above three relations into \eqref{eq:total-variations}, we get 
\begin{align}\nonumber
2\|\nu^{i'}-\mu^{i'}\|_{TV}=\frac{B_2-B_1}{A+B}+\frac{C_1-C_2}{A+C}+\frac{A|B-C|}{(A+C)(A+B)}.
\end{align}
Let us define $u=\max_{i,x}c_i(x)$ and note that $\Delta_{x_i}\leq u$ and $\Delta_{y_i}\leq u$. Then, $A+B\ge ke^{-\beta u}$ and $A+C\ge ke^{-2\beta u}$. Using the mean-value theorem for $f(r)=e^{-\beta r}$, we have the following relations:
\begin{align}\nonumber
&B_2-B_1\leq \beta \big|c_{i'}(x_i,x_{-i'})-c_{i'}(y_i,x_{-i'})\big|e^0\leq \beta u,\cr
&C_1-C_2\leq \beta \big|c_{i'}(x_i,x_{-i'})-\Delta_{x_i}-c_{i'}(y_i,x_{-i'})-\Delta_{y_i}\big|e^{\beta u}\leq 3\beta ue^{\beta u}\cr 
& |B-C|\leq |B_1-C_1|+|B_2-C_2|\leq 2\beta ue^{\beta u}.
\end{align} 
Therefore, for any $i'$ we have
\begin{align}\label{eq:total-variation-bound}
2\|\nu^{i'}-\mu^{i'}\|_{TV}\leq \frac{\beta u}{k}e^{\beta u}+\frac{3\beta u}{k}e^{3\beta u}+\frac{2\beta u}{k}e^{2\beta u}\leq \frac{6\beta u}{k}e^{3\beta u}\ \ \Rightarrow \ \|\nu^{i'}-\mu^{i'}\|_{TV}\leq \frac{3\beta u}{k}e^{3\beta u}. 
\end{align}
Next, we bound the mixing time of the Glauber dynamics. For one step of the chain, we have
\begin{align}\label{eq:one-step-chain}
\mathbb{E}[\rho(Z^x_1,Z^y_1)]=1-\frac{1}{n}+\frac{1}{n}\sum_{i'\neq i}\|\mu^{i'}-\nu^{i'}\|_{TV}.
\end{align}
Substituting \eqref{eq:total-variation-bound} into \eqref{eq:one-step-chain} and using the assumption of $\beta\leq \frac{k}{6nu}$, we get 
\begin{align}\nonumber
\mathbb{E}[\rho(Z^x_1,Z^y_1)]\leq 1-\frac{1}{n}+\frac{n-1}{nk}3\beta u e^{3\beta u}\leq 1-\frac{1}{n}+\frac{1}{k}(3\beta u)e^{3\beta u}\leq 1-\frac{1}{n}+\frac{1}{2n}e^{\frac{k}{4n}}\leq 1-\frac{1}{7n}.
\end{align}
Starting from any two arbitrary initial states $x$ and $y$ that differ in $d$ positions, we can reach from $x$ to $y$ using a sequence $x^0=x, x^1,\ldots, x^d=y$ such that every two consecutive allocations differ in exactly one position. As $\rho(\cdot,\cdot)$ is metric over the space of allocations, using triangle inequality, we can write
\begin{align}\nonumber
\mathbb{E}[\rho(Z^{x}_1,Z^{y}_1)]\leq \sum_{k=1}^{d}\mathbb{E}[\rho(Z^{x^{k-1}}_1,Z^{x^{k}}_1)]\leq (1-\frac{1}{7n})d=(1-\frac{1}{7n})\rho(x,y). 
\end{align}
Moreover, using the Markov property of the chains, we can write 
\begin{align*}
\mathbb{E}[\rho(Z^{x}_t,Z^{y}_t)]=\mathbb{E}\big[\mathbb{E}[\rho(Z^x_t,Z^y_t)|Z^x_{t-1},Z^y_{t-1}]\big]=\mathbb{E}[\rho(Z^{Z^x_{t-1}}_1,Z^{Z^y_{t-1}}_1)]\leq (1-\frac{1}{7n})\mathbb{E}[\rho(Z^x_{t-1},Z^y_{t-1})].  
\end{align*}
By using the above inequality recursively, we obtain
\begin{align*}
\mathbb{E}[\rho(Z^{x}_t,Z^{y}_t)]\leq (1-\frac{1}{7n})^t\rho(Z^x_{0},Z^y_{0})=(1-\frac{1}{7n})^t\rho(x,y)\leq ne^{-\frac{t}{7n}}.
\end{align*}
Finally, using Markov's inequality, we can write
\begin{align*}
\mathbb{P}(Z^{x}_t\neq Z^{y}_t)\leq \mathbb{P}(\rho(Z^{x}_t, Z^{y}_t)\ge 1)\leq \mathbb{E}[\rho(Z^{x}_t, Z^{y}_t)]\leq ne^{-\frac{t}{7n}}.
\end{align*}
The above relation, in view of Lemma \ref{lemm:mixing}, shows that the mixing time of the Gluaber dynamics is at most $t_{\rm mix}(\epsilon)=O(n\ln\frac{n}{\epsilon})$.
\end{proof}

As we mentioned earlier, there is a trade-off between the mixing time of the Glauber dynamics and the stationary distribution induced by them. For the Glauber dynamics to concentrate around the allocation profile with a globally minimum potential function, one must choose a large noise parameter $\beta$. However, choosing a large $\beta$ can result in a slow mixing time (as opposed to the fast mixing time guarantee given in Theorem \ref{thm:mixing} for smaller values of $\beta$). However, in practice, the costs of players are mainly determined by their nearby neighbors. Thus, one can leverage the locality of players' cost functions to establish a fast mixing time of the Glauber dynamics for larger values of $\beta$. For instance, if we know that each player's action can affect the cost of at most $d$ nearby players, then the bound for $\beta$ in the previous theorem can be improved to $\beta=O(\frac{1}{u}\ln(\frac{k}{d}))$.

\section{Auction-Based Distributed Computation for Data Placement Problem}

In the previous section, we developed a distributed game-theoretic framework that allows players to update their resources selfishly subject to some noise and eventually be able to compute the global minimizer of the data placement problem. In that formulation, players are the agents, and the actions are the choices of resources. An alternative perspective is to view the resources as players who bid to buy the cache spaces of the agents (viewed as items). That leads us to the following auction-based distributed algorithm for the data placement problem.  

\subsection{An Auction-Based Distributed Algorithm}

Let us consider an auction with $k$ players (resources) and $n$ items. We view the unit cache space of agent $i$ as an item that will be sold to players. We assume that each player $\ell\in [k]$ represents a set of $n$ clients $\{(j,\ell): j\in [n]\}$ and charges a $\beta_j^{\ell}\ge 0$ fee per unit demand to client $(j,\ell)$. This charge is for representing client $(j,\ell)$ in the auction and for connecting that client to resource $\ell$. Moreover, we assume that the items are sold separately using a first-price auction in which players submit their bids for different items. An item is sold to the player with the highest bid (ties are broken arbitrarily), and the winner must pay an amount equal to the highest bid. In addition, we assume that the entrance fee for player $\ell$ to participate in the auction for item $i$ is $f_i^{\ell}$. Next, we specify the bidding strategy for the players.

Let us consider player $\ell$, who charges $\beta_j^{\ell}$ per unit demand to its client $(j,\ell)$. From that amount, player $\ell$ subtracts $c_{ij}$ to account for the cost of connecting $(j,\ell)$ to agent $i$, and therefore includes only a $(\beta_j^{\ell}-c_{ij})^+$ portion of that amount toward bidding for item $i$. Therefore, summing over the total demand of all the clients, player $\ell$ bids $\big(\sum_j w_j^{\ell}(\beta_j^{\ell}-c_{ij})^+-f^{\ell}_i\big)^+$ toward item $i$, where the term $f_i$ is to account for the entrance fee that player $\ell$ has to pay to be able to bid for item $i$. Thus, if player $\ell$ wins a bundle of items $X^{\ell}\subseteq [n]$ in the auction, $\ell$'s utility equals the amount that $\ell$ collects from its clients minus the amount that $\ell$ has to pay to the auctioneer, i.e., 
\begin{align}\label{eq:u_x}
u_{\ell}(\beta^{\ell}, X^{\ell})&=\sum_{j}w_j^{\ell}\beta_{j}^{\ell}-\sum_{i\in X^{\ell}}\big(\sum_j w_j^{\ell}(\beta_j^{\ell}-c_{ij})^+-f^{\ell}_i\big)^+\cr 
&=\sum_{j}w_j^{\ell}\beta_{j}^{\ell}-\sum_{i\in X^{\ell}}\big(\sum_j w_j^{\ell}(\beta_j^{\ell}-c_{ij})^+-f^{\ell}_i\big),
\end{align}
where the second equality holds by individual rationality. Otherwise, if $\sum_j w_j^{\ell}(\beta_j^{\ell}-c_{ij})^+-f^{\ell}_i<0$ for some $i$, there is no incentive for player $\ell$ to enter the auction for item $i$. Therefore, player $\ell$'s goal is to determine a charging strategy $\beta_j^{\ell}$ to maximize its utility function.   

\subsection{Performance Guarantee of the Solution} 

To analyze the performance guarantee of the allocation profile obtained from the above auction, let us again consider the dual program corresponding to the LP relaxation of \eqref{eq:primal} given by
\begin{align}\label{eq:agai_dual}
\max &\sum_{j,\ell}w_j^{\ell}\beta_j^{\ell}-\sum_i \alpha_i\cr
& \sum_j w_j^{\ell}\big(\beta_j^{\ell}-c_{ij}\big)^+-f^{\ell}_i\leq \alpha_i   \ \forall i,\ell,\cr
& \beta^{\ell}_{j}, \alpha_i \ge 0, \ \forall i,j,\ell.
\end{align} 
To satisfy all the constraints in \eqref{eq:agai_dual} while maximizing the dual objective function, we must set $\alpha_i=\max_{\ell}\big(\sum_j w_j^{\ell}\big(\beta_j^{\ell}-c_{ij}\big)^+-f^{\ell}_i\big)^+$, which gives us the following compact form for the dual program:
\begin{align}\label{eq:compact_dual}
\max_{\beta^{\ell}_{j} \ge 0} \Big\{\sum_{j,\ell}w_j^{\ell}\beta_j^{\ell}-\sum_i \max_{\ell}\big(\sum_j w_j^{\ell}\big(\beta_j^{\ell}-c_{ij}\big)^+-f^{\ell}_i\big)^+\Big\}.
\end{align}
Let us use variables $y_i^{\ell}$ to represent the inner maximization in \eqref{eq:compact_dual} as\footnote{In fact, $y_i^{\ell}$ can be thought of as dual variables corresponding to the first set of constraints in the convex program \eqref{eq:agai_dual}, which also coincide with the primal variables $y_i^{\ell}$ in the original linear program \eqref{eq:primal}.} 
\begin{align}\label{eq:minmax_dual}
\max_{\beta^{\ell}_{j} \ge 0}\min_{\substack{\sum_{\ell}y^{\ell}\leq \boldsymbol{1}\\ y^{\ell}\ge \boldsymbol{0} \forall \ell}} \Big\{\sum_{j,\ell}w_j^{\ell}\beta_j^{\ell}-\sum_i\sum_{\ell} y_i^{\ell}\big(\sum_j w_j^{\ell}\big(\beta_j^{\ell}-c_{ij}\big)^+-f^{\ell}_i\big)\Big\}.
\end{align}
Now, for every resource type $\ell$, let us define a utility function as 
\begin{align}\nonumber
u_{\ell}(\beta^{\ell}, y^{\ell})=\sum_{j}w_j^{\ell}\beta_{j}^{\ell}-\sum_{i}y_i^{\ell}\big(\sum_j w_j^{\ell}(\beta_j^{\ell}-c_{ij})^+-f^{\ell}_i\big),
\end{align}
which is the same as the utility function \eqref{eq:u_x} defined for player $\ell$ if we take $X^{\ell}=\{i: y_i^{\ell}=1\}$. Then, the dual program \eqref{eq:minmax_dual} can be written as
\begin{align}\label{eq:minmax_dual_u}
\max_{\beta^{\ell}_{j} \ge 0}\min_{\substack{\sum_{\ell}y^{\ell}\leq \boldsymbol{1}\\ y^{\ell}\ge \boldsymbol{0} \forall \ell}} \sum_{\ell} u_{\ell}(\beta^{\ell}, y^{\ell})=\min_{\substack{\sum_{\ell}y^{\ell}\leq \boldsymbol{1}\\ y^{\ell}\ge \boldsymbol{0} \forall \ell}}\max_{\beta^{\ell}_{j} \ge 0} \sum_{\ell} u_{\ell}(\beta^{\ell}, y^{\ell}),
\end{align}
where the equality holds because each $u_{\ell}(\beta^{\ell}, y^{\ell})$ is concave in $\beta^{\ell}$ and linear (convex) in $y^{\ell}$. In particular, we note that the optimal solution to $y$ is always integral because $u_{\ell}(\beta^{\ell}, y^{\ell})$ is linear with respect to $y^{\ell}$ and the constraints $\{\sum_{\ell}y^{\ell}\leq \boldsymbol{1}, y^{\ell}\ge \boldsymbol{0} \forall \ell\}$ define an integral polytope.

Using the above derivations, it should be clear that if $(\beta,y)$ is an optimal dual solution to \eqref{eq:minmax_dual_u}, then player $\ell$'s strategy to maximize its utility is to charge $w_j^{\ell}\beta^{\ell}_j$ to client $(j,\ell)$, and $\ell$ receives item $i$ if $y^{\ell}_i=1$, in which case $\ell$ has to pay $\alpha_i=\sum_j w_j^{\ell}(\beta_j^{\ell}-c_{ij})^+-f^{\ell}_i$, which is the maximum bid among all the bids for item $i$. Therefore, the allocation profile obtained from the auction when all the players selfishly maximize their utilities is the same as the optimal dual solution to the min-max problem \eqref{eq:minmax_dual_u}.

\begin{theorem}\label{thm:auction}
Consider the data placement problem with zero placement costs $f^{\ell}_i=0,\forall i,\ell$, and let $(\alpha,\beta)$ be the optimal solution to the dual program \eqref{eq:dual}. Then, the resource allocation profile obtained from the auction is an $(\frac{1}{1-\gamma})$-approximation of the data placement problem, where $\gamma=\frac{\|\alpha\|_1}{\|\beta\|_1}\in [0,1)$.
\end{theorem}
\begin{proof}
We prove the theorem through the following four steps:

{\bf I) Primal feasibility and integrality:} By abuse of notation, let $(\beta,y)$ be the minimal optimal solution to the min-max dual problem \eqref{eq:minmax_dual_u} (i.e., a solution with the least number of nonzero entries), and let $X^{\ell}=\{i: y_i^{\ell}=1\}$. Clearly, $(X^{\ell}, \ell\in [k])$ partitions the set of agents $[n]$. We complement this solution with an integral feasible solution for the primal program by setting $x^{\ell}_{ij}=1$ if $i= \argmin_{i'\in X^{\ell}}c_{i'j}$ (ties are broken arbitrarily), i.e., we connect agent $j$ to the closest agent in $X^{\ell}$ to get access to resource $\ell$. Then, $(x,y)$ defined in this way forms a feasible integral solution to the primal program \eqref{eq:primal}. 

{\bf II) Dual feasibility:} Let $\beta$ be the optimal solution to \eqref{eq:minmax_dual_u} and define $\alpha_i=\max_{\ell}\big(\sum_j w_j^{\ell}(\beta_j^{\ell}-c_{ij})^+-f^{\ell}_i\big)^+$ and $u^{\ell}_{ij}=w_j^{\ell}\big(\beta_j^{\ell}-c_{ij}\big)^+$. Then, from the above arguments, $(\{\alpha_i\},\{w_j^{\ell}\beta_j^{\ell}\},\{u^{\ell}_{ij}\})$ forms an optimal solution to the dual program \eqref{eq:dual}.

{\bf III) Complementary slackness conditions:} Based on Proposition \ref{prop_NP-hard}, we cannot expect all the complementary slackness conditions for the above primal-dual solutions to hold.\footnote{Otherwise, the proposed primal solution will be an optimal integral solution to an NP-hard problem.} However, as we show, the proposed solutions satisfy most of these conditions and still constitute a good suboptimal solution. 
\begin{itemize}
\item[1)] Since $y$ is a minimal optimal integral solution to \eqref{eq:minmax_dual}, $y_i^{\ell}=1$ implies that 
\begin{align}\nonumber
\alpha_i=\max_{\ell'}\big(\sum_j w_j^{\ell'}(\beta_j^{\ell'}-c_{ij})^+-f^{\ell'}_i\big)^+=\big(\sum_j w_j^{\ell}(\beta_j^{\ell}-c_{ij})^+-f^{\ell}_i\big)^+=\sum_j w_j^{\ell}(\beta_j^{\ell}-c_{ij})^+-f^{\ell}_i,
\end{align}
where the last equality holds because if $\sum_j w_j^{\ell}(\beta_j^{\ell}-c_{ij})^+-f^{\ell}_i<0$ we would have $y_i^{\ell}=0$.
\item[2)] Since the primal solution $(X^{\ell}, \ell\in [k])$ partitions $[n]$, and using the definition of $x$ that assigns every agent $j$ to exactly one agent in each $X^{\ell}$, the primal constraints $\sum_{\ell}y_i^{\ell}\leq 1$ and $\sum_{i}x^{\ell}_{ij}\ge 1$ are always satisfied by equality. Therefore, the complementary slackness conditions always hold for these two types of primal constraints.    
\item[3)] As we showed before, for the optimal dual solution we have $u_{ij}^{\ell}=w_j^{\ell}\big(\beta_j^{\ell}-c_{ij}\big)^+ \forall i,j,\ell$. Therefore, to show complementary slackness for the first set of dual constraints in \eqref{eq:dual}, we only need to show that if $x_{i'j'}^{\ell}=1$ for some $i',j',\ell$, then $\big(\beta_{j'}^{\ell}-c_{i'j'}\big)^+=\beta_{j'}^{\ell}-c_{i'j'}$, or, equivalently, $\beta_{j'}^{\ell}\ge c_{i'j'}$. This is also true because if $x^{\ell}_{i'j'}=1$, that means $y_{i'}^{\ell}=1$ (and thus $i'\in X^{\ell}$) and $c_{i'j'}\leq c_{ij'}, \forall i\in X^{\ell}$. Moreover, using case (1) we have $\alpha_{i}=\sum_j w_j^{\ell}(\beta_j^{\ell}-c_{ij})^+-f^{\ell}_{i}\ge 0, \forall i\in X^{\ell}$. Therefore, 
\begin{align}\nonumber
u_{\ell}(\beta^{\ell}, X^{\ell})&=\sum_{j}w_j^{\ell}\beta_{j}^{\ell}-\sum_{i\in X^{\ell}}\big(\sum_j w_j^{\ell}(\beta_j^{\ell}-c_{ij})^+-f^{\ell}_i\big)\cr
&=\sum_{i\in X^{\ell}}f^{\ell}_i+\sum_{j}w_j^{\ell}\beta_{j}^{\ell}-\sum_j \sum_{i\in X^{\ell}} w_j^{\ell}(\beta_j^{\ell}-c_{ij})^+\cr 
&=\sum_{i\in X^{\ell}}f^{\ell}_i+\sum_{j\neq j'}\Big(w_j^{\ell}\beta_{j}^{\ell}-\sum_{i\in X^{\ell}} w_j^{\ell}(\beta_j^{\ell}-c_{ij})^+\Big)+ \Big(w_{j'}^{\ell}\beta_{j'}^{\ell}-\sum_{i\in X^{\ell}} w_{j'}^{\ell}(\beta_{j'}^{\ell}-c_{ij'})^+\Big).
\end{align} 
Suppose, by contrary, $\beta_{j'}^{\ell}<  c_{i'j'}$. Then $\beta_{j'}^{\ell}<c_{ij'} \forall i\in X^{\ell}$ and we have $\sum_{i\in X^{\ell}} w_{j'}^{\ell}(\beta_{j'}^{\ell}-c_{ij'})^+=0$. Therefore, if $\beta_{j'}^{\ell}$ is slightly increased, the last term in the above expression strictly increases.\footnote{Note that since the dual function is given by the sum of utilities defined over separate variables $\beta^{\ell}, \ell\in [k]$, such an increase does not affect other terms in the dual objective function.} This contradicts the fact that $\beta^{\ell}$ corresponds to the optimal dual solution that maximizes $u_{\ell}(\cdot, X^{\ell})$.
\end{itemize}

{\bf IV) Bounding the performance:} Using the properties of the primal-dual solutions that we established above, the only set of constraints that may violate the complementary slackness conditions are the primal constraints $x^{\ell}_{ij}\leq y^{\ell}_{i}$ with the corresponding dual variables $u^{\ell}_{ij}=w_j^{\ell}\big(\beta_j^{\ell}-c_{ij}\big)^+$. Therefore, using Lemma \ref{lemm:block-complement} with the block of constraints $A_1\tilde{x}\ge b_1$ representing constraints $y^{\ell}_{i}-x^{\ell}_{ij}\ge 0$ and corresponding dual variables $u^*_1=(u^{\ell}_{ij})$, the cost of the generated primal solution denoted by $\mbox{Cost}(X)$ equals
\begin{align}\label{eq:cost-bound-alpha-f}
\mbox{Cost}(X)&=\mbox{OPT}+\sum_{i,j,\ell}u^{\ell}_{ij}(y_i^{\ell}-x^{\ell}_{ij})\leq \mbox{OPT}+\sum_{i,j,\ell}u^{\ell}_{ij}(1-0)\cr 
&= \mbox{OPT}+\sum_{\ell}\sum_{i\in X^{\ell}}\sum_{j}w_j^{\ell}\big(\beta_j^{\ell}-c_{ij}\big)^+\cr 
&=\mbox{OPT}+\sum_{\ell}\sum_{i\in X^{\ell}}(\alpha_i+f_i^{\ell})\cr 
&= \mbox{OPT}+\sum_i\alpha_i,
\end{align}
where the third equality holds through use of case 1 of the complementary slackness conditions, and the last equality holds by the assumption $f^{\ell}_i=0,\forall i,\ell$. Dividing both sides by $\mbox{OPT}=\sum_{j,\ell}\beta_j^{\ell}-\sum_i\alpha_i$,\footnote{Here, we are using the original definition of dual variables $\beta_j^{\ell}$ given in \eqref{eq:dual} rather that their scaled version $w_j^{\ell}\beta_j^{\ell}$.} and using the definition of $\gamma$, completes the proof. 
\end{proof}

\begin{remark}
In fact, the optimal dual objective value denoted by $\mbox{OPT}$ equals the social welfare resulting from the auction, i.e., the sum of the players' utilities $\mbox{SW}:=\sum_{\ell}u_{\ell}=\sum_{j,\ell}w_j^{\ell}\beta_j^{\ell}-\sum_i \alpha_i$, while the revenue derived by the auctioneer equals to the sum of all the payments $\mbox{Rev}:=\sum_{i}\alpha_i$. Therefore, another way of interpreting the result of Theorem \eqref{thm:auction} is to say that the approximation guarantee of the allocation profile obtained from the auction is $1+\frac{\mbox{Rev}}{\mbox{SW}}$.   
\end{remark}     

\section{Conclusions}
In this paper, we studied the general non-metric data placement problem from a multiagent game-theoretic perspective and devised distributed computation algorithms for obtaining or approximating its global optimal solutions. The motivation behind this work is that in many real-world applications, such as peer-to-peer systems or ad hoc storage systems, the servers are selfish entities that only want to maximize their own payoffs, yet the goal is to achieve good global performance in terms of content delivery and resource availability. We showed that although the problem is hard to approximate within a logarithmic factor, some natural Glauber dynamics can converge to the optimal solution for a sufficiently large noise parameter. In particular, we established a fast mixing time of the dynamics to their stationary distribution for a certain range of noise parameters. We also provided an auction-based algorithm that can approximate the optimal global solution and can be easily implemented in a distributed manner.

\section*{Appendix I}

\begin{lemma}\label{lemm:mixing}
Let $Z^x_t$ and $Z^y_t$ be copies of a Markov chain with initial states $x$ and $y$ and transition probability matrix $P$. Suppose that for each pair of initial states $x, y\in \mathcal{X}$ there is a coupling $(Z^x_t, Z^y_t)$. Then, $d(t)\leq \max_{x,y}\mathbb{P}(Z^x_t\neq Z^y_t)$, where $d(t)=\sup_{\mu}\|\mu P^t-\pi\|_{TV}$ is the maximum total variation between the distribution of the Markov chain at time $t$ and its stationary distribution $\pi$. In particular, the mixing time of the Markov chains is at most 
\begin{align}\nonumber
t_{\rm mix}(\epsilon):=\min\{t: d(t)<\epsilon\}\leq \min\{t: \max_{x,y}\mathbb{P}(Z^x_t\neq Z^y_t)<\epsilon\}.     
\end{align}
\end{lemma}
\begin{proof}
The proof follows from Theorem 5.4 and Corollary 5.5 in \cite{levin2017markov}.
\end{proof}

\begin{lemma}\label{lemm:block-complement}
Consider an LP: $\mbox{OPT}=\min\{cx: Ax\ge b, x\ge 0\}$ and its dual $\max\{ub: uA\leq c, u\ge 0\}$. Suppose $A=[\frac{A_1}{A_2}]$ can be represented using two blocks of constraints $A_1$ and $A_2$. Let $u^*$ be the optimal dual solution, and assume $\tilde{x}$ is a feasible primal solution such that $(\tilde{x},u^*)$ satisfy all the complementary slackness conditions except for the constraints $A_1\tilde{x}\geq b_1$ with corresponding dual variables $u^*_1$, where $b=[\frac{b_1}{b_2}]$. Then $\tilde{x}$ forms an approximate optimal solution for the LP such that $c\tilde{x}= \mbox{OPT}+u_1^*(A_1\tilde{x}-b_1)$.  
\end{lemma}
\begin{proof}
Since dual constraints satisfy complementary slackness with respect to $\tilde{x}$, $(u^*A-c)\tilde{x}=0$. Moreover, since all the primal constraints except $A_1\tilde{x}\geq b_1$ satisfy complementary slackness conditions,
\begin{align}\nonumber
u^*(A\tilde{x}-b)=u^*_1(A_1\tilde{x}-b_1)+u^*_2(A_2\tilde{x}-b_2)=u^*_1(A_1\tilde{x}-b_1) \ \Rightarrow \ u^*A\tilde{x}=u^*b+u^*_1(A_1\tilde{x}-b_1).
\end{align}
Thus, we conclude that $\tilde{x}$ is a feasible primal solution whose objective cost equals
\begin{align}\nonumber
c\tilde{x}=u^*A\tilde{x}=u^*b+u^*_1(A_1\tilde{x}-b_1)=\mbox{OPT}+u^*_1(A_1\tilde{x}-b_1),
\end{align}
where the last equality follows by strong duality. 
\end{proof}

\bibliographystyle{IEEEtran}
\bibliography{thesisrefs}

\begin{thebibliography}{10}
\providecommand{\url}[1]{#1}
\csname url@samestyle\endcsname
\providecommand{\newblock}{\relax}
\providecommand{\bibinfo}[2]{#2}
\providecommand{\BIBentrySTDinterwordspacing}{\spaceskip=0pt\relax}
\providecommand{\BIBentryALTinterwordstretchfactor}{4}
\providecommand{\BIBentryALTinterwordspacing}{\spaceskip=\fontdimen2\font plus
\BIBentryALTinterwordstretchfactor\fontdimen3\font minus
  \fontdimen4\font\relax}
\providecommand{\BIBforeignlanguage}[2]{{%
\expandafter\ifx\csname l@#1\endcsname\relax
\typeout{** WARNING: IEEEtran.bst: No hyphenation pattern has been}%
\typeout{** loaded for the language `#1'. Using the pattern for}%
\typeout{** the default language instead.}%
\else
\language=\csname l@#1\endcsname
\fi
#2}}
\providecommand{\BIBdecl}{\relax}
\BIBdecl

\bibitem{guo2016algorithmic}
M.~Guo, ``Algorithmic mechanism design for data replication problems,'' Ph.D.
  dissertation, University of Cincinnati, 2016.

\bibitem{gopalakrishnan2012cache}
R.~Gopalakrishnan, D.~Kanoulas, N.~N. Karuturi, C.~Pandu~Rangan, R.~Rajaraman,
  and R.~Sundaram, ``Cache me if you can: {C}apacitated selfish replication
  games,'' in \emph{Latin American Symposium on Theoretical Informatics}.\hskip
  1em plus 0.5em minus 0.4em\relax Springer, 2012, pp. 420--432.

\bibitem{etesami2020complexity}
S.~R. Etesami, ``Complexity and approximability of optimal resource allocation
  and {N}ash equilibrium over networks,'' \emph{SIAM Journal on Optimization},
  vol.~30, no.~1, pp. 885--914, 2020.

\bibitem{etesami2017price}
S.~R. Etesami and T.~Ba{\c{s}}ar, ``Price of anarchy and an approximation
  algorithm for the binary-preference capacitated selfish replication game,''
  \emph{Automatica}, vol.~76, pp. 153--163, 2017.

\bibitem{angel2014optimal}
E.~Angel, E.~Bampis, G.~G. Pollatos, and V.~Zissimopoulos, ``Optimal data
  placement on networks with a constant number of clients,'' \emph{Theoretical
  Computer Science}, vol. 540, pp. 82--88, 2014.

\bibitem{baev2001approximation}
I.~D. Baev and R.~Rajaraman, ``Approximation algorithms for data placement in
  arbitrary networks,'' in \emph{Proceedings of the Twelfth Annual ACM-SIAM
  Symposium on Discrete Algorithms}, 2001, pp. 661--670.

\bibitem{baev2008approximation}
I.~Baev, R.~Rajaraman, and C.~Swamy, ``Approximation algorithms for data
  placement problems,'' \emph{SIAM Journal on Computing}, vol.~38, no.~4, pp.
  1411--1429, 2008.

\bibitem{swamy2016improved}
C.~Swamy, ``Improved approximation algorithms for matroid and knapsack median
  problems and applications,'' \emph{ACM Transactions on Algorithms (TALG)},
  vol.~12, no.~4, pp. 1--22, 2016.

\bibitem{krishnaswamy2018constant}
R.~Krishnaswamy, S.~Li, and S.~Sandeep, ``Constant approximation for $k$-median
  and $k$-means with outliers via iterative rounding,'' in \emph{Proceedings of
  the 50th Annual ACM SIGACT Symposium on Theory of Computing}, 2018, pp.
  646--659.

\bibitem{ansari2017large}
L.~Ansari, \emph{Large-scale {O}ptimization for {D}ata {P}lacement
  {P}roblem}.\hskip 1em plus 0.5em minus 0.4em\relax University of Lethbridge
  (Canada), 2017.

\bibitem{thakral2017approximation}
S.~Thakral, \emph{Approximation Algorithms for Data Placement Problems}.\hskip
  1em plus 0.5em minus 0.4em\relax University of Delhi, 2017.

\bibitem{drwalcompetitive}
M.~Drwal, ``Competitive algorithms for online data placement on uncapacitated
  uniform network,'' in \emph{Proceedings of the The Fifth International
  Conference on Advances in Future Internet}, 2013, pp. 1--7.

\bibitem{jain2003greedy}
K.~Jain, M.~Mahdian, E.~Markakis, A.~Saberi, and V.~V. Vazirani, ``Greedy
  facility location algorithms analyzed using dual fitting with
  factor-revealing {LP},'' \emph{Journal of the ACM (JACM)}, vol.~50, no.~6,
  pp. 795--824, 2003.

\bibitem{charikar1999constant}
M.~Charikar, S.~Guha, {\'E}.~Tardos, and D.~B. Shmoys, ``A constant-factor
  approximation algorithm for the $k$-median problem,'' in \emph{Proceedings of
  the Thirty-First Annual ACM Symposium on Theory of Computing}, 1999, pp.
  1--10.

\bibitem{deng2022constant}
S.~Deng, ``Constant approximation for fault-tolerant median problems via
  iterative rounding,'' \emph{Operations Research Letters}, vol.~50, no.~4, pp.
  384--390, 2022.

\bibitem{drwal2014decomposition}
M.~Drwal and J.~Jozefczyk, ``Decomposition algorithms for data placement
  problem based on {L}agrangian relaxation and randomized rounding,''
  \emph{Annals of Operations Research}, vol. 222, no.~1, pp. 261--277, 2014.

\bibitem{hochbaum1982heuristics}
D.~S. Hochbaum, ``Heuristics for the fixed cost median problem,''
  \emph{Mathematical Programming}, vol.~22, no.~1, pp. 148--162, 1982.

\bibitem{bienkowski2020nearly}
M.~Bienkowski, B.~Feldkord, and P.~Schmidt, ``A nearly optimal deterministic
  online algorithm for non-metric facility location,'' \emph{arXiv preprint
  arXiv:2007.07025}, 2020.

\bibitem{williamson2011design}
D.~P. Williamson and D.~B. Shmoys, \emph{The {D}esign of {A}pproximation
  {A}lgorithms}.\hskip 1em plus 0.5em minus 0.4em\relax Cambridge University
  Press, 2011.

\bibitem{levin2017markov}
D.~A. Levin and Y.~Peres, \emph{Markov {C}hains and {M}ixing {T}imes}.\hskip
  1em plus 0.5em minus 0.4em\relax American Mathematical Soc., 2017, vol. 107.

\end{thebibliography}

\end{document}